\documentclass[preprint,11pt]{article}

\usepackage{amssymb,amsfonts,amsmath,amsthm,amscd,dsfont,mathrsfs}
\usepackage{graphicx,float,psfrag,epsfig,color}

\RequirePackage{bbm}
\newtheorem{theorem}{{Theorem}}

\newtheorem{definition}{\textbf{Definition}}%[section]
\newtheorem{corollary}{\textbf{Corollary}}%[section]

\allowdisplaybreaks

\def\E{{\mathbb E}}

\def\normal{{\sf N}}

\newcommand{\SortNoop}[1]{}
\newcommand {\bE} {\mathbb{E}}

\newcommand {\hx} {\widehat{x}}

\newcommand {\bW} {\overline{W}}

\newcommand {\ind}{\mathbbm{1}}
\newcommand {\indicator}[1]{\ind_{\{ #1 \}}}

\newcommand {\reals}{\mathbb{R}}
\newcommand {\complex}{\mathbb{C}}
\newcommand {\integers}{\mathbb{Z}}
\newcommand {\ham} {{\mathcal H}}

\newcommand {\matB} {{\mathbf{B}}}

\newcommand {\matC} {{\mathbf{C}}}
\newcommand {\matA} {{\mathbf{A}}}
\newcommand {\matD} {{\mathbf{D}}}
\newcommand {\matone} {{1}}

\newcommand {\matW} {\mathbf{W}}

\newcommand {\trace} {\text{\rm Tr}}
\newcommand {\matI} {{\mathbf{I}}}
\newcommand {\inp} {{\sf in}}

\def\de{{\rm d}}
\def\argmin{{\rm argmin}}
\def\<{\langle}
\def\>{\rangle}
\def\cX{{\cal X}}

\begin{document}

\title{Applications of Lindeberg Principle in Communications
and Statistical Learning}

\author{Satish Babu Korada\thanks{Department of Electrical Engineering,
Stanford University} \;\;\; and\;\;\;
Andrea Montanari${}^{*,}$\thanks{Department of Statistics,
Stanford University}}

\date{}

\maketitle

\begin{abstract}
We use a generalization of the Lindeberg principle developed 
by Sourav Chatterjee to prove universality properties for various problems in
communications, statistical learning and random matrix theory.
We also show that these systems can be viewed as the  
limiting case of a properly defined sparse system.
The latter result is useful when the sparse systems are easier to analyze than
their dense counterparts. The list of problems we
consider is by no means exhaustive. We believe that the ideas can 
be used in many other problems relevant for information theory.
\end{abstract} 

\section{Introduction}\label{sec:introduction}

The phenomenon of universality is common to many disciplines of science and
engineering. A well known example is the central limit theorem which, 
in a simple version, says the following. 
Let $\{X_i\}_{i\ge 1}$ be a collection of i.i.d.
random variables with mean zero and variance  $\bE[X_i^2] = 1$. Then 
\begin{align*}
\frac{1}{\sqrt n}\sum_{i=1}^n X_i \stackrel{{\rm d}}{\rightarrow}\normal(0,1),
\end{align*}
where $\stackrel{{\rm d}}{\rightarrow}$ denotes convergence in distribution 
as $n\to\infty$, and
$\normal(0,1)$ is a Gaussian random variable with mean zero and
variance one. In particular, the central limit theorem implies 
that the distribution of $n^{-1/2}\sum_{i=1}^n X_i$ is asymptotically
independent of the details of the distribution of the summands $X_i$.  
In other words, its limit  is ``universal'' for a large class of 
summands' distributions. 
Other examples include the limiting spectrum of random matrices
\cite{MarPas}, and various properties of statistical mechanics models
\cite{StatMech}. 

Examples in communications theory where universal properties have been
established include the MIMO communications problem \cite{Tel99}. 
In
these problems it was shown the capacity of the system is independent of the
distribution of the fading coefficients and the spreading sequences
respectively.

A different research area in which universality ideas 
appear ubiquitous is compressed sensing. 
Donoho and Tanner \cite{DoTUniversality} carried out a systematic
empirical investigation of universality in this context. In particular
they showed that the phase transition boundary in the sparsity-undersampling
tradeoff is universal for a large class of sensing matrices. 
The precise location of this phase transition was determined earlier on
in the case of Gaussian sensing matrices \cite{DoT09}.

A related phenomenon which we study here is the sparse-dense equivalence. As an 
example consider a uniformly random regular graph $G_n$ of degree $d$ 
over $n$ vertices. Let $A_n\in\reals^{n\times n}$ be the symmetric matrix 
whose non-vanishing entries correspond to edges in $G_n$ and take
values in $\{+1/\sqrt{d},-1/\sqrt{d}\}$ independently and uniformly at random.
As $n\to\infty$ the spectral measure of such a matrix converges
almost surely \cite{Kesten,McKay} to a well defined limit $\rho_d(\de\lambda)$
supported on $[-2\sqrt{1-1/d},2\sqrt{1-1/d}]$, where:
\begin{eqnarray}
\rho_d(\de\lambda) = \frac{1}{2\pi}\, 
\frac{\sqrt{4(1-1/d)-\lambda^2}}{1-\lambda^2/d}\, \de\lambda\, .
\end{eqnarray}
If we now consider the $d\to\infty$ limit, this
distribution converges weakly to the celebrated semi-circle law
\begin{eqnarray}
\rho_\infty(\de\lambda) = \frac{1}{2\pi}\, 
\sqrt{4-\lambda^2}\, \de\lambda\, .
\end{eqnarray}
This is the limiting spectrum of the standard (dense) Wigner matrices. 
We refer to this type of property as to a sparse-dense equivalence. Showing 
such a relationship can be particularly useful
when the analysis of the sparse system is easier than its dense
counterpart. Specific examples will be provided below. 

Universality and sparse-dense equivalence can have far reaching consequences
in communications and information theory.
In this paper, we demonstrate this by studying both phenomena 
within a common framework, and obtaining  new results 
in each of the above mentioned problems. The main tool that we use is the
following generalization of Lindeberg's principle that was proved in
\cite{Sou06}. 
%
%**********************************************************************
%
\subsection{Lindeberg Principle}

Given $f:\reals^n\to\reals$, the generalized Lindeberg
principle provides conditions under which the distribution of 
$f(X_1,\dots,X_n)$ is approximately insensitive
to the distribution of its arguments $X_1,\dots,X_n$
which are assumed to be independent. This generalizes
the classical Lindeberg proof of the central limit theorem, that focused
on $f(x_1,\dots,x_n) = (x_1+\dots +x_n)/\sqrt{n}$. 

Let us restate here the main result of \cite{Sou06}.
\begin{theorem}[Generalized Lindeberg Principle, 
\cite{Sou06}]\label{thm:sourav}
Let $U=(U_1,\dots,U_n)$ and $V = (V_1,\dots,V_n)$ be two random 
vectors with mutually independent components. For $1\leq i\leq
n$, define 
\begin{align*}
a_i &\equiv |\bE[U_i] - \bE[V_i]|,\\
b_i &= |\bE[U_i^2] - \bE[V_i^2]|.
\end{align*}
and further assume $\max_i(\bE\{|U_i|^3\} + \bE\{|V_i|^3\})\le M_3$. Suppose
$f:\reals^n\to \reals$ is a thrice continuously differentiable 
function, and for
$r=1,2,3$, let $L_r(f)$ be a finite constant such that $|\partial^r_i f(u)|\leq
L_r(f)$ for each $i$ and $u\in\reals^n$, 
where $\partial_i^r$ denotes the $r$-fold
derivative in the $i$th coordinate. Then
\begin{align*}
|\bE[f(U)]-\bE[f(V)]| \leq \sum_{i=1}^n(a_i L_1(f) + \frac{1}{2}b_i L_2(f)) +\frac{1}{6}nL_{3}(f)M_3.
\end{align*}
\end{theorem}
Notice that, while this theorem explicitly bounds the change 
in expectation $f(\,\cdot\, )$, it gives control on its distribution
as well, by applying it to $g(f(\, \cdot\, ))$, for $g:\reals\to\reals$ 
belonging to a suitable class of test functions.
 
In many problems of interest for this paper,
the bound on the derivatives of $f$ required by the last
theorem does not hold,   and  a more careful analysis is needed.
For that purpose we use the following theorem.
The proof is analogous to the one of Theorem~\ref{thm:sourav},
and is  provided in Section~\ref{sec:Proof1}.
\begin{theorem}\label{thm:souravmodified}
Let $U=(U_1,\dots,U_n)$ and $V = (V_1,\dots,V_n)$ be two random 
vectors with mutually independent components.
Let $\{a_i\}_{1\le i\le n}$ and $\{b_i\}_{1\le i\le n}$ be as defined in 
Theorem~\ref{thm:sourav}. Then
\begin{align*}
|\bE[f(U)] - \bE[f(V)]|\leq \sum_{i=1}^n\Big\{& a_i \bE[|\partial_i
f(U_1^{i-1},0,V_{i+1}^n)|] + \frac{1}{2}b_i
\bE[|\partial_i^2f(U_1^{i-1},0,V_{i+1}^n)|]\\
&+ \frac{1}{2}\bE\int_0^{U_i} \big[|\partial^3_{i}f(U_1^{i-1},s,V_{i+1}^n)|\big](U_i-s)^2\,
\de s\\ 
& + \frac{1}{2}\bE\int_0^{V_i} \big[|\partial^3_{i}f(U_1^{i-1},s,V_{i+1}^n)|\big](V_i-s)^2\, \de s
\Big\}\, .
\end{align*}
\end{theorem}
%
%********************************************************************
%
\section{Applications}\label{sec:applications}

In this section we discuss the application of Theorem 
\ref{thm:souravmodified} to a problem from communications theory
(code division multiple access channels), 
and one from statistical learning theory 
(estimation via LASSO). We also revisit a standard model from statistical
mechanics (the Sherrington-Kirkpatrick model), and 
the spectrum of Wishart matrices, which is related to capacity of MIMO 
channels.
In each of these cases, Theorem \ref{thm:souravmodified} 
implies both universality and sparse-dense equivalence results. 
We will not try to be exhaustive, but rather to point out some 
selected conclusion. This Section contains definitions and statements, 
while proofs are deferred to section \ref{sec:Proof2}.

In the following, we use uppercase letters, e.g, $X,Y$, to 
denote random variables and their
lowercase counterparts, e.g. $x,y$, to denote realizations
of such random variables.
We also use boldface characters to denote random matrices,
with the subscript to indicate their dimension,
e.g. $\matA_n$, $\matB_n$.

Most of our results concern random matrices with i.i.d. 
entries and apply under some simple centering and
normalization conditions, provided the entries have finite 
sixth moment. Rather than repeating these conditions at
each of the results below, we introduce them once and for all.
\begin{definition}[Random Matrices of Standard Type] 
Let $\matA_n = \{A_{ij}\}_{1\le i\le m,1\le j\le n}$ be a 
sequence of random matrices indexed by their dimensions $m$ and $n$
(with $m=m_n$ an appropriate sequence of integers).
We say that $\matA_n$ is a random matrix of \emph{standard type}
if the entries $\{A_{ij}\}_{i,j\ge 1}$ form an array 
of independent and identically distributed random variables with 
$\bE[A_{ij}]=0$, $\bE[A^2_{ij}]=1$ and 
$\bE[A_{ij}^6]\le {K} <\infty$, for some $K$ independent of $m,n$.
\end{definition}
%
%**************************************************************
%
\subsection{Capacity of a CDMA System}\label{sec:cdma}

Code Division Multiple Access (CDMA) is a widely used communication system
between multiple users and a common receiver
\cite{VerduBook}. The scheme
consists of $n$ users modulating their information
sequence by a signature sequence (spreading sequence) of length
$m$ and transmitting the resulting signal. The number $m$ is sometimes referred
to as the spreading gain or the number of chips per sequence. The receiver
obtains the sum of all transmitted signals and the noise which is 
often assumed to be white and Gaussian (AWGN). 

For the sake of simplicity, we will assume antipodal signals:
each user wishes to communicate a symbol $X_k \in \{ + 1,-1\}$, 
to the common receiver. 
User $k$ uses a signature
sequence $(A_{1k},\dots,A_{m k})$, with $A_{ik}\in\reals$.
The received signal $Y_i$ in the $i$-th time interval is given by 
\begin{equation*}
Y_i= \sum_{k=1}^n A_{ik}\, X_k +  \sigma\, Z_i\, ,
\end{equation*}
where $Z_i$ are i.i.d. copies of $\normal(0,1)$
and therefore the noise power is $\sigma^2$.

We use $x^{\inp}=(x_1,\dots,x_n)$ to denote any specific realization 
of the transmitted symbols, and will assume that a realization
of such symbols is used uniformly at random. The corresponding
random vector is $X^{\inp}= (X_1,\dots,X_n)$ while 
$Y=(Y_1,\dots,Y_m)$ is the received signal. Typically $X^{\inp}$ is chosen to be
uniformly distributed over $\{+1,-1\}^n$. In this paper we restrict to this
case. However it is possible to generalize the results below 
to a large class of distributions for the symbol $X_i$.

We write $A_n$ for the $m\times n$ matrix 
$\{A_{ik}\}_{1\le i\le m,1\le j\le n}$. 
Let $C_n(\matA_n)$ denote the capacity of such system, i.e.
the number of bits per user that can be reliably transmitted
to the common receiver under the above constraints. Explicitly we have
\begin{align}\label{capacity}
C_n(A_n)&=\log 2 -\frac{1}{2}\alpha
-\frac{1}{n}\bE_Y\log\Big\{\sum_{x\in\{+1,-1\}^n}e^{-\frac{1}{2\sigma^2}\Vert
Y-A_n x\Vert_2^2}\Big\}\, .
\end{align}
Here expectation $\E_Y$ is taken over the received signal. 

Random spreading sequences were initially considered in \cite{GrA96}. 
Here, the signature sequences are modeled as random vectors
with i.i.d. components $\{A_{ik}\}_{1\le i\le m, 1\le k\le n}$.
Without loss of generality we can assume $\E\{A_{ik}\}= 0$
and $\E\{A_{ik}^2\}=1$. We will be interested in 
the large system limit $m,n\to \infty$ with $\alpha=m/n$
fixed.

In order to keep the average power (per symbol) equal to $1$,
we will rescale the signature matrix by a factor $1/n$. 
For a random signature matrix $\matA_n$, we consider 
therefore the capacity $C_n(m^{-1/2}\matA_n)$, which is itself random. 
As proved in \cite{KoM08}, $C_n(m^{-1/2}\matA_n)$ does in fact 
concentrate exponentially around its expectation. This motivates us to
focus on its expectation.
\begin{theorem}[Universality of the Capacity of random CDMA
sytems]\label{thm:cdmauniversal}
Let $\matA_n = \{A_{ij}\}_{1\le i\le m,1\le j\le n}$ and 
$\matB_n = \{B_{ij}\}_{1\le i\le m,1\le j\le n}$ denote two 
$m\times n$ dimensional random  spreading matrices  of standard type.  Then
\begin{align*}
\lim_{n\to\infty, m=n\alpha}\big\{\bE[C_n(m^{-1/2}\matA_n)] -\bE[C_n(m^{-1/2}\matB_n)]
\big\} = 0\, .
\end{align*}
\end{theorem}
The above theorem establishes that the per-user capacity of a CDMA 
channel is asymptotically independent 
of the distribution of the spreading sequences. 
The conditions required to be satisfied
by the distributions are milder than the ones imposed in \cite{KoM08}.

Our next result concerns the sparse-dense equivalence.
Sparse signature schemes were proposed in  \cite{MoT06itw}
both as a tool for simplifying mathematical analysis and 
as a design option with potential practical advantages.
Given a signature matrix $\matA=\{A_{ij}\}_{1\le i\le m,
1\le j\le n}$ defined as above, 
its sparsification $\matA^{\gamma}_n$ is given by
\begin{eqnarray}
A^{\gamma}_{ij} =\left\{\begin{array}{ll}
A_{ij} & \mbox{with probability\; $\gamma/n$,}\\
0  & \mbox{with probability\; $1-\gamma/n$,}
\end{array}\right.
\end{eqnarray}
with $\gamma>0$ a design parameter that is kept fixed in the large system 
limit.
Under a sparse signature scheme, the power per symbol is normalized to
$1$ if we rescale the signatures by a factor $1/\sqrt{\gamma}$.
The channel output is therefore
\begin{equation*}
Y_i= \frac{1}{\sqrt{\gamma}}\sum_{k=1}^n A_{ik}^{\gamma}\, X_k +  \sigma\, Z_i
\, .
\end{equation*}
We can then prove the following sparse-dense equivalence result.
\begin{theorem}[Sparse-Dense Equivalence for CDMA channels]
\label{thm:cdmasparsedense}
Let $\matA_n = \{A_{ij}\}_{1\le i\le m,1\le j\le n}$ and 
$\matB_n = \{B_{ij}\}_{1\le i\le m,1\le j\le n}$ denote two 
$m\times n$ dimensional random  spreading matrices  of standard type.
For $\gamma>0$, let $\matA_n^{\gamma}$ be the sparsification of 
$\matA_n$.
 Then
\begin{align*}
\lim_{\gamma\to\infty}\lim_{n\to\infty,m=n\alpha}
\big\{\bE[C_n(\gamma^{-1/2}\matA^\gamma_n)] -\bE[C_n(n^{-1/2}\matB_n)]\big\} = 0 .
\end{align*} 
\end{theorem}

As already mentioned, establishing sparse-dense equivalence is
particularly useful when the analysis of a sparse system is 
simpler than for its dense counterpart.
In \cite{MoT06itw} it was shown that there exists 
$\alpha_{\rm s}>0$ such that, for all $\alpha \leq \alpha_s$,  
\begin{align}\label{eqn:MoTformula}
\lim_{\gamma\to\infty}\lim_{n\to\infty,m=n\alpha}\bE[C_n(\gamma^{-1/2}\matA_n^{\gamma})]
= \min_{m\in[0,1]}C_{\rm RS}(q)\, ,
\end{align}
where 
\begin{eqnarray}
C_{\rm RS}(q) & = &
\frac{\lambda}{2} (1+q) -\frac{1}{2\alpha}\log \lambda\sigma^2 -
{\sf E}_z\{ 
\log(2\cosh(\sqrt{\lambda} Z+\lambda))\}
\, ,\\
\lambda & \equiv &\frac{1}{\sigma^2+\alpha(1-q)}\, ,
\end{eqnarray}
where ${\sf E}_z$ denotes expectation with respect to $Z\sim \normal(0,1)$.
The parameter $\alpha_{\rm s}$ is defined as
the largest $\alpha$ such that the maximizer in 
\eqref{eqn:MoTformula} is unique. Numerically $\alpha_{\rm s}\approx 1.49$. 
The same formula was derived earlier by Tanaka \cite{Tanaka}
using the non-rigorous replica method from statistical physics.

Combining this with
Theorem~\ref{thm:cdmasparsedense} we can conclude the following result for the
capacity of a random CDMA system.
\begin{corollary}[Capacity of random CDMA systems]
Let $\matA_n$ denote an $m\times n$ dimensional random spreading
matrix with i.i.d. entries.
Assume $\bE[A_{ij}] = 0$, $\bE[A_{ij}^2] = 1$ and $\bE[A_{ij}^6]\leq K <\infty$. 
Then for $\alpha\leq\alpha_{\rm s}$
\begin{align*}
\lim_{n\to\infty,m=n\alpha}\E[C_n(m^{-1/2}\matA_n)] = 
\min_{q\in[0,1]}C_{\rm RS}(q).
\end{align*} 
\end{corollary}
%
%**********************************************************
%
\subsection{Estimation via LASSO}\label{sec:lasso}

The LASSO (also known as basis pursuit de-noising) is 
a popular strategy in statistical learning, used for
reconstructing high-dimensional parameter vectors from noisy measurements
\cite{Tibshirani,ChenDonoho}. 
It is particularly well suited when the underlying parameters 
vector is sparse in an appropriate basis. For this very reason,
it is object of intense study within the compressed sensing literature.

We assume here that a signal $x_0\in\reals^n$ is observed through
the sensing matrix $A_n$ which has dimensions 
$m \times n$. The measurements $y\in \reals^m$ are modeled as a noisy 
linear functions
\begin{eqnarray} 
y = A_n \,x_0 + z\, ,
\end{eqnarray}
with $z\in\reals^m$ a noise vector.
 Let the noise vector $z$ be i.i.d. 
Gaussian vector. The recovery of $x_0$ from $y$ is done using the
following convex optimization problem
\begin{align}
\hx(\lambda) = \argmin_{x\in\reals^n} \Big\{\frac{1}{2}\Vert y - A_n 
x\Vert_2^2+\lambda\, \Vert x\Vert_1 \Big\}\, .
\end{align}
For some applications the sensing matrix $A_n$
is not far from random or pseudo-random. It is important to
ask to which degree results obtained for a specific distribution 
of $A_n$ generalize to other distributions \cite{DoT09,DoTUniversality}.
We consider the case in which the entries $x_{0,i}$
of $x_0$  are uniformly bounded, i.e.,
$|x_{0i}| \leq x_{\rm max}$ 
for some constant $x_{\rm max} > 0$ independent of $n,m$. 
We further assume that the noise vector $z$ has i.i.d. 
entries $z_i\sim\normal(0,\sigma^2)$ and focus on the limit
$m,n\to\infty$ with $m/n=\alpha$ fixed.

The next result provides rigorous evidence towards the broader
universality picture, by proving universality for the 
normalized cost
\begin{eqnarray}
L(A_n) = \frac{1}{n}\, \min_{x\in[-x_{\rm max},x_{\rm max}]^n}\,
\Big\{\frac{1}{2}\Vert y 
- A_n x\Vert_2^2+\lambda\, \Vert x\Vert_1 \Big\}\, .
\end{eqnarray}
\begin{theorem}[Universality for LASSO]\label{thm:lassouniversal}
Let $\matA_n = \{A_{ij}\}_{1\le i\le m,1\le j\le n}$ and 
$\matB_n = \{B_{ij}\}_{1\le i\le m,1\le j\le n}$ denote two 
$m\times n$ dimensional random  sensing matrices  of standard type. 
Then
\begin{align*}
\lim_{n\to\infty,\, m=n\alpha} 
\big\{\bE[L(n^{-1/2}\matA_n)]-\bE[L(n^{-1/2}\matB_n)]\big\}=0\, .
\end{align*}
\end{theorem}
%
%**********************************************************
%
\subsection{Spectrum of Wishart matrices and capacity of MIMO channels}
\label{sec:wishart}

Given an $n\times n$  symmetric matrix $W_n$, let 
$\{\lambda_i(W_n)\}_{1\le i\le n}$ denote its eigenvalues. The
spectral measure of $W_n$ is the probability measure
\begin{eqnarray}
\mu_n \equiv \frac{1}{n}\sum_{i=1}^{n}\delta_{\lambda_i(W_n)}\, .
\end{eqnarray} 
The study of the limit of $\mu_n$ as $n\to\infty$, for a sequence
of random matrices $\matW_n$ is a central topic in random matrix theory,
with important applications in multi-antenna communications.
A well-studied example is the family of Wishart matrices. 
Here, $\matW_m = \frac{1}{n}\matA_n^{\top}\matA_n$, where
$\matA_n$ is an $m\times n$ matrix, whose entries are i.i.d. realizations of a
zero mean random variable with variance $1$.

A standard approach to characterizing the spectral measure 
is through its Stieltjes transform \cite{Guionnet}
which is defined as
\begin{align*}
S_n(\matW_n,z) = \frac{1}{n}\sum_{i=1}^n\frac{1}{z+\lambda_i(\matW_n)} =
\frac{1}{n}\trace\big((\matW_n+zI_n)^{-1}\big),
\end{align*}
where $z\in\complex\backslash\reals$ and $I_n$ is 
the $n$-dimensional identity
matrix.
The limiting spectrum of the family $\{\matW_n\}_{n\ge 1}$ 
can be obtained by computing
$\lim_{n\to\infty}S_n(\matW_n,z)$. The universality of Wishart matrices 
is a well known result 
\cite{MarPas}. The following is a sparse-dense equivalence
result for this class of matrices.
\begin{theorem}[Sparse-Dense Equivalence for Wishart Matrices]
\label{thm:wishartsparsedense}
Let $\matA_n = \{A_{ij}\}_{1\le i\le m,1\le j\le n}$ and 
$\matB_n = \{B_{ij}\}_{1\le i\le m,1\le j\le n}$ denote two 
$m\times n$ dimensional random  matrices  of standard type.
For $\gamma>0$, let $\matA_n^{\gamma}$ be the sparsification of 
$\matA_n$. Let $\matW_{\matA,n}^{\gamma}\equiv \gamma^{-1}
(\matA^\gamma_n)^{\top}\matA^\gamma_n$ and 
$\matW_{\matB,n}\equiv n^{-1}
(\matB_n)^{\top}\matB_n$. 
Then for all $z\in\complex\backslash\reals$
\begin{align*}
\lim_{\gamma\to\infty}\lim_{n\to\infty,m=n\alpha}
\big\{\bE[S_n(\matW_{\matA,n}^{\gamma},z)] -\bE[S_n(\matW_{\matB,n},z)]\big\} = 0\, .
\end{align*} 
\end{theorem}

Under appropriate tightness conditions, 
convergence of Stieltjes transforms implies 
weak convergence of the spectrum $\mu_n$,
which further implies the convergence of the empirical average
$\frac{1}{n}\sum_i f(\lambda_i)$ for any continuous bounded function $f$.
As a particular application of this remark, we consider
the capacity of multi-input multi-output
(MIMO) communication systems. The channel model is very similar to the CDMA
system discussed in Section~\ref{sec:cdma}. 
For a channel input $X=(X_1,\dots,X_n)$, the channel output 
is a vector $Y = (Y_1,\dots,Y_m)$ in $\reals^m$, with components
\begin{align*}
Y_i = \sum_{k=1}^n H_{ik} X_k + \sigma\, Z_i
\end{align*}
where $Z_i$ are i.i.d. realizations of $\normal(0,1)$.
However, in this case it is customary to not restrict the inputs
to be $\{+1,-1\}$, but rather to impose a power constraint
$n^{-1}\sum_{i=1}^n\E\{X_i^2\}\le 1$.
Given a channel gains matrix $H_n=\{H_{ij}\}_{1\le i\le m,1\le j\le n}$,
the average capacity per input antenna \cite{Tel99} is then given by
\begin{align*}
C_n(H_n) = \max_{\{Q\succeq 0: \frac{1}{n}\sum_{i=1}^nQ_{ii} =
1\}}\frac{1}{2n}\bE\Big\{\log\,{\rm Det} \Big(\matI_m +\frac{1}{\sigma^2}
H_n Q H_n^{\top}
\Big)\Big\}\, .
\end{align*}
when the input covariance is $Q$
For the case of $H_{ij}$ being i.i.d. symmetric Gaussian random variables it was shown in
\cite{Tel99} that the above maximum is achieved for $Q=\matI_n$. Here, we 
assume that little is known about the channel gains and therefore
this covariance matrix is used for other matrices $H_{n}$ as well. Under this
assumption, the achievable average rate is given by 
\begin{align*}
C_n(H_n) = \frac{1}{2n}\sum_{i=1}^{m} \log\Big\{1+
\frac{1}{\sigma^2}\lambda_i(H_n
H_n^{\top})\Big\} = \frac{1}{2n}\sum_{i=1}^{n}
\log\Big\{1+\frac{1}{\sigma^2}\lambda_i(H_n^\top
H_n)\Big\}.
\end{align*}
Under the above theorem implies the following result for the MIMO channels.
\begin{corollary}[Sparse-Dense Equivalence for the MIMO Capacity]\label{cor:mimosparsedense}
Let $\matA_n = \{A_{ij}\}_{1\le i\le m,1\le j\le n}$ and 
$\matB_n = \{B_{ij}\}_{1\le i\le m,1\le j\le n}$ denote two 
$m\times n$ dimensional random  matrices  of standard type.
For $\gamma>0$, let $\matA_n^{\gamma}$ be the sparsification of 
$\matA_n$. Then
\begin{align*}
\lim_{\gamma\to\infty}\lim_{n\to\infty,m=n \alpha }
\big\{\bE[C_n(\gamma^{-1/2}\matA^\gamma_n)] -
\bE[C_n(n^{-1/2}\matB_n)]\big\} = 0\, .
\end{align*} 
\end{corollary}

%
%****************************************************************
%
\subsection{Spin glass models}\label{sec:sk}

Spin glass models have been object of intense interest within 
statistical mechanics, mathematical physics and probability
theory. Both rigorous and heuristic techniques from this 
domain have been applied with success in information theory \cite{MM}.

A number of universality and sparse-dense equivalence results
have been proved in this context \cite{Sou06,Tala,GuerraToninelli}. 
We re-derive two of these results here because they provide a very
simple and instructive illustration of the proof technique that is used 
in the more intricate examples listed in the previous sections.

We  focus in particular on the Sherrington-Kirkpatrick (SK) model.
The model is defined by the Hamiltonian function
$\ham:\{+1,-1\}^n\times \reals^{n\times n}\to\reals$ given by
\begin{align*}
\ham(x,A_n) = -\frac1{\sqrt 2}\sum_{i,j=1}^{n}A_{ij}x_ix_j =-\,\frac{1}{\sqrt
2} x^{\top}A_n x\, ,
\end{align*}
for an $n\times n$ dimensional matrix $A_n$ and $x=(x_1,\dots,x_n)
\in\{+1,-1\}^n$.
An important object of interest in this context is the 
free entropy density  at inverse temperature $\beta$,
which is defined by
\begin{align*}
f(\beta,A_n) \equiv \frac{1}{n}\log
\Big\{\sum_{x\in\{+1,-1\}^n}e^{-\beta\ham(x,A_n)}\Big\}.
\end{align*}

Universality of the free energy for the SK model 
was established in \cite{Tal02} and was later extended to general distributions in \cite{Ton02}.
As shown in \cite{Sou06} the current approach gives a stronger result.
\begin{theorem}[Universality for the SK model \cite{Sou06}]\label{thm:skuniversal}
Let $\matA_n = \{A_{ij}\}_{1\le i,j\le n}$ and 
$\matB_n= \{B_{ij}\}_{1\le i,j\le n}$ be two $n\times n$ 
dimensional random  matrices. Assume 
that both $\{A_{ij}\}$ and $\{B_{ij}\}$ are collections
of i.i.d. random
variables with $\bE[A_{ij}]=\bE[B_{ij}]=0$,
$\bE[A^2_{ij}]=\bE[B^2_{ij}]=1$, and $\bE[|A_{ij}|^3],
\bE[|B_{ij}|^3]\leq K <\infty$. Then
\begin{align*}
\lim_{n\to\infty}\big\{\bE[f(\beta,n^{-1/2}\matA_n)]-\bE[f(\beta,n^{-1/2}\matB_n)]\big\}=0\,.
\end{align*}
\end{theorem}

The sparse-dense equivalence was proved in \cite{GuerraToninelli}
under the slightly stronger assumption of uniformly bounded entries
$|A_{ij}|\le 1$ with even distribution.
\begin{theorem}[Sparse-Dense Equivalence]\label{thm:sksparsedense}
Let $\matA_n = \{A_{ij}\}_{1\le i,j\le n}$ and 
$\matB_n= \{B_{ij}\}_{1\le i,j\le n}$ be two $n\times n$ 
dimensional random  matrices. Assume 
that both $\{A_{ij}\}$ and $\{B_{ij}\}$ are collections
of i.i.d. random
variables with $\bE[A_{ij}]=\bE[B_{ij}]=0$,
$\bE[A^2_{ij}]=\bE[B^2_{ij}]=1$, and $\bE[|A_{ij}|^3],
\bE[|B_{ij}|^3]\leq K <\infty$. 
For $\gamma>0$, let $\matA_n^{\gamma}$ be the sparsification of $\matA_n$. 
Then
\begin{align*}
\lim_{\gamma\to\infty}\lim_{n\to\infty}
\big\{\bE[f(\beta,\gamma^{-1/2}\matA_n^\gamma)] -\bE[f(\beta,n^{-1/2}\matB_n)]\big\}=0\, .
\end{align*} 
\end{theorem}
%
%*****************************************************************
%
\section{Proof of Theorem~\ref{thm:souravmodified}}\label{sec:proofSou}
\label{sec:Proof1}

\begin{proof}[Proof of Theorem~\ref{thm:souravmodified}]
Let $\partial^r_i f$ denote $\frac{\partial^r f}{\partial x^r_i}$. Let 
$$\bW_i= (U_1,\dots,U_{i},V_{i+1},\dots,V_n),$$
$$\bW^0_i=(U_1,\dots,U_{i-1},0,V_{i+1},\dots,V_n).$$ Then 
\begin{align}\label{eqn:telescoping}
\bE[f(U)] - \bE[f(V)] = \sum_{i=1}^n (\bE[f(\bW_i)] - \bE[f(\bW_{i-1})])\, .
\end{align}
From the third-order Taylor expansion, we have
\begin{align}\label{eqn:taylor2}
f(\bW_i) = f(\bW_i^0) + U_i\partial_i f(\bW_i^0) +
\frac{U_i^2}{2}\partial^2_if(\bW_i^0) + 
\frac{1}{2}\int_0^{U_i}\partial_i^3f(U_1^{i-1},s,V_{i+1}^n)(U_i-s)^2
\de s\, .
\end{align}
Similarly, we get
\begin{align}\label{eqn:taylor1}
f(\bW_{i-1}) = f(\bW_i^0) + V_i\partial_i f(\bW_i^0) +
\frac{V_i^2}{2}\partial^2_if(\bW_i^0) + 
\frac{1}{2}\int_0^{V_i}\partial_i^3f(U_1^{i-1},s,V_{i+1}^n)(V_i-s)^2
\de s.
\end{align}
From Eq.~\eqref{eqn:telescoping},  using \eqref{eqn:taylor2}
and \eqref{eqn:taylor1}, we
get
\begin{align*}
&\bE[f(U)] - \bE[f(V)]
=\sum_{i=1}^n\Big\{\bE[(U_i-V_i)\partial_if(\bW_i^0)] +
\frac12\bE[(U_i^2-V_i^2)\partial_i^2f(\bW_i^0)]\\
&+\bE\big[\frac{1}{2}\int_0^{U_i}\partial_i^3f(U_1^{i-1},s,V_{i+1}^n)(U_i-s)^2
\de s\big]
+\bE\big[\frac{1}{2}\int_0^{V_i}\partial_i^3f(U_1^{i-1},s,V_{i+1}^n)(V_i-s)^2
\de s\big]\Big\}.
\end{align*}
The result follows by noting that $f(\bW_i^0)$ is independent of 
$\{U_i,V_i\}$.
\end{proof}
%
%*****************************************************************
%
\section{Proofs of statements from Section \ref{sec:applications}}
\label{sec:Proof2}

We will present the proofs starting from the last 
example, i.e. the Sherrington-Kirkpatrick model in Section 
\ref{sec:sk}. As mentioned, this is a particularly simple example of 
the general proof strategy. 
%
%*****************************************************************
%
\subsection{SK Model}\label{sec:skproof}

As mentioned in
Section~\ref{sec:sk}, the Hamiltonian for this model is given by 
\begin{align*}
\ham(x,A_n) = -\frac{1}{\sqrt 2}\, x^\top A_nx,
\end{align*}
where $A_n$ is an $n\times n$ dimensional matrix.
For a function $(x,A_n)\mapsto g(x,A_n)$, we denote by 
$\<g(x,A_n)\>$ its expectation with respect to the 
probability distribution $p_{A_n}(x) \propto \exp\{-\beta\ham(x,A_n)\}$
on $\{+1,-1\}^n$. Explicitly:
\begin{align}
\langle g(x,A_n)\rangle = \frac{\sum_{x \in\mathcal{X}^n} g(x,A_n)\,
e^{-\ham(x,Z,A_n)}}{\sum_{x\in\mathcal{X}^n} e^{-\ham(x,Z,A_n)}}\, .
\end{align}
Denote by $\partial^{k}_{rc}$ the $k$-th partial derivative 
with respect to $A_{rc}$ (row $r$, column $c$).
A straightforward calculation shows that 
third derivative $\partial^3_{rc}f(\beta,A_n)$ is given by 
\begin{align*}
\partial^3_{rc}f(\beta, A_ n) =\frac{\beta^3}{\sqrt{2}n}
\langle x_rx_r \rangle(1-\langle
x_rx_c\rangle^2)\ ,
\end{align*}
which implies $L_3(f) \leq \beta^3/(\sqrt{2}n)$ (with $L_3$ defined as in 
Theorem \ref{thm:sourav}). 

\begin{proof}[Proof of Theorem~\ref{thm:skuniversal}]
From the definition of the random matrices $\matA_n$ and $\matB_n$, we have 
we have $\bE[A_{ij}] = \bE[B_{ij}]$, $\bE[A_{ij}^2]= \bE[B_{ij}^2]$ and 
$\bE[|A_{ij}|^3]\leq (1+K)$, $\bE[|B_{ij}|^3]\leq (1+K)$.
Using Theorem~\ref{thm:sourav} we get
\begin{align*}
|\bE[f(n^{-1/2}\matA_n)] - \bE[f(n^{-1/2}\matB_n)]| \leq
\frac{1}{6}n^2\frac{\beta^3}{\sqrt{2}\, n}\max_{r,c\in [n]}\Big\{\bE\Big[\frac{|A_{rc}|^3}{n^{3/2}}\Big],
\bE\Big[\frac{|B_{rc}|^3}{n^{3/2}}\Big]\Big\}=O\Big(\frac{1}{\sqrt n}\Big).
\end{align*}
\end{proof}
\begin{proof}[Proof of Theorem~\ref{thm:sksparsedense}]
From the definition of the random matrices 
$\matA^\gamma_n$ and $\matB_n$, we have
we have $\bE[A_{ij}] = \bE[B_{ij}]$, $\bE[A_{ij}^2]= \bE[B_{ij}^2]$
and $\bE[|A^{\gamma}_{ij}|^3]\leq (1+K)\gamma/n$, $\bE[|B_{ij}|^3]\leq (1+K)$
(with $K$ independent of $\gamma$ and $n$). Therefore 
using the estimate on $L_3(f)$ fro the previous proof,
together with  Theorem~\ref{thm:sourav} we have
\begin{align*}
|\bE[f(\gamma^{-1/2}\matA^{\gamma}_n)] - \bE[f(n^{-1/2}\matB_n)]| \leq
\frac{1}{6}n^2\frac{\beta^3}{\sqrt{2}\, n}
\max_{r,c\in [n]}\Big\{\bE\Big[\frac{|A^\gamma_{rc}|^3}{\gamma^{3/2}}\Big],
\bE\Big[\frac{|B_{rc}|^3}{n^{3/2}}\Big]\Big\}
\leq K'\beta^3 \max\Big\{\frac{1}{\sqrt\gamma},\frac{1}{\sqrt n}\Big\}.
\end{align*}
Therefore, $\lim_{\gamma\to\infty}\lim_{n\to\infty}\big\{\bE[f(\gamma^{-1/2}\matA^\gamma_n)]-\bE[f(n^{-1/2}\matB_n)]\big\} =0.$
\end{proof}
%
%*********************************************************************
%
\subsection{CDMA}\label{sec:cdmaproof}

For any $m\times n$ matrix $A_n$, the capacity \eqref{capacity} can be expressed as 
\begin{align*}
C_n(A_n) = \log 2
-\frac{1}{2}\alpha-\frac{1}{n}\sum_{x^{\inp}\in\{+1,-1\}^n}\frac{1}{2^n}\bE_{Z}\log\Big\{\sum_{x\in\{+1,-1\}^n}
e^{-\frac{1}{2\sigma^2}\Vert Z + A_n(x^{\inp} - x)\Vert^2_2}\Big\}\, ,
\end{align*}
where $Z$ is an $m$-dimensional random vector, whose entries are i.i.d.
$\normal(0,\sigma^2)$. By a simple change of variables in the sum over
$x$, we get
\begin{align*}
C_n(A_n) = \log 2
-\frac{1}{2}\alpha-\frac{1}{n}\sum_{x^{\inp}\in\{+1,-1\}^n}\frac{1}{2^n}\bE_{Z}\log\Big\{\sum_{x\in\{0,2\}^n}
e^{-\frac{1}{2\sigma^2}\Vert Z + A_nx^{\inp}x\Vert^2_2}\Big\}\, .
\end{align*}
For a matrix $A_n=\{A_{i,j}\}_{1\le i\le m,1\le j\le n}$, and a vector 
$x^{\inp}\in\{+1,-1\}^n$, define
$A_n(x^{\inp})$ by letting $[A_n(x^{\inp})]_{ij} = A_{ij}x^{\inp}_j$.
Further, define the Hamiltonian function
$\ham:\{0,2\}^n\times\reals^m\times\reals^{m\times n}\to\reals$ 
by
\begin{align*}
\ham(x,Z,A_n) = \frac{1}{2\sigma^2}\Vert Z +  A_nx\Vert_2^2 =
\frac{1}{2\sigma^2}\sum_{i=1}^m\Big(Z_i + \sum_{j=1}^n A_{ij}x_j\Big)^2\, .
\end{align*}
Then we have
\begin{eqnarray*}
C_n(A_n) & = &\log 2
-\frac{1}{2}\alpha- \frac{1}{2^n}\sum_{x^{\inp}\in\{+1,-1\}^n}\E_Z
f(A_n(x^{\inp}),Z)\, ,\\
f(A_n,Z) & \equiv &  \frac{1}{n}\log
\Big\{\sum_{x\in\{0,2\}^n} e^{-\ham(x,Z,A_n)}\Big\}\, .
\end{eqnarray*}
If $\matA_n$ is a random matrix of standard type,
and $x^{\inp}\in\{+1,-1\}^n$, then $\matA_n(x^{\inp})$ is also a random matrix 
of standard type. In order to prove the universality 
results, theorems \ref{thm:cdmauniversal} and \ref{thm:cdmasparsedense},
it is therefore sufficient to fix --say-- $x^{\inp} =
(+1,\dots,+1)$, and prove universality of $\E_Z f(A_n,Z)$.

Analogously to the proof in the previous
section, for a function $(x,Z,A_n)\mapsto g(x,Z,A_n)$, we let
\begin{align}\label{eqn:gibbs}
\langle g(x,Z,A_n)\rangle \equiv \frac{\sum_{x \in\mathcal\{0,2\}^n} g(x,Z,A_n)
e^{-\ham(x,Z,A_n)}}{\sum_{x\in\mathcal{X}^n} e^{-\ham(x,Z,A_n)}}\, .
\end{align}
In order
use Theorem~\ref{thm:souravmodified} we need to estimate the third derivatives
of $f$. Again, $\partial_{rc}^k f$ denote the $k$-th derivative of 
$f$ with respect
to the $A_{r,c}$. 
The third derivative is then given by 
\begin{align}
\partial_{rc}^3f(A_n,Z) = \frac{1}{n(2\sigma^2)^3}\Big(- \langle
(\partial_{rc}\ham(x,Z,A_n))^3\rangle + 3 \langle
\partial_{rc}\ham(x,Z,A_n)\rangle
\langle
(\partial_{rc}\ham(x,Z,A_n))^2\rangle
- 2 \langle
\partial_{rc}\ham(x,Z,A_n)\rangle^3\Big)\, .\nonumber\\
\label{eq:FreeEnergyDerivative}
\end{align}

\begin{proof}[Proof of Theorem~\ref{thm:cdmauniversal}]
Let $\matA_n$ and $\matB_n$ be as defined in the theorem.
Let $\matD_n(r,c,s)$ denote the matrix with entries
\begin{align*}
D_{ij} = \left\{\begin{array}{ll}
\frac{1}{\sqrt m}A_{ij},& \text{ if } i< r \text{ or } i=r \text{ and } j< c,\\
s,& \text{ if } i=r, \text{ and } j=c,\\
\frac{1}{\sqrt m}B_{ij},& \text{ otherwise}.
\end{array}\right.
\end{align*}
From now onwards we use $\ham(x)$ to denote $\ham(x,Z,\matD_n(r,c,s))$
and let $\<\,\cdot\,\>$ denote the corresponding average,
as per Eq.~(\ref{eqn:gibbs}). 
Further, for $r\in [m]$, let 
$\Theta_{r}(x) \equiv (Z_r + \sum_{j=1}^n D_{rj}x_{j})/(\sqrt 2\sigma)$ and 
$\ham_{\sim r}(x) =\ham(x)-\Theta_r(x)^2$.
Notice that 
\begin{eqnarray*}
\ham(x) = \sum_{i\in [m]}\Theta_i(x)^2\, ,\;\;\;\;\;\;\;\; 
 \ham_{\sim r}(x)  = \sum_{i\in [m]\setminus r}\Theta_i(x)^2\, .
\end{eqnarray*}
Accordingly, we let $\langle \,\cdot\,\rangle_{\sim r}$ denote the average as defined in
\eqref{eqn:gibbs} with the Hamiltonian
$\ham_{\sim r}(x)$.

The derivative of $\ham(x)$ with respect to $A_{rc}$ is 
\begin{align*}
\partial_{rc}\ham(x) = \frac{1}{2\sigma^2}\Big(Z_r + \sum_{j=1}^n
D_{rj}x_{j}\Big) 2{x_c}= \frac{1}{\sqrt 2\sigma}2x_c\Theta_r(x).
\end{align*}
Its fourth moment can then be bounded as 
\begin{align*}
\bE\langle (\partial_{rc}\ham)^4\rangle_s & =
\bE\Big\{\frac{\sum_{x}e^{-\ham(x)}(\partial_{rc}\ham)^4}{\sum_{x}e^{-\ham(x)}}
\Big\}\\
& \leq \bE \Big\{\frac{\sum_{x}e^{-\ham_{\sim r}(x)
-\Theta_r(x)^2}(64/\sigma^4)\Theta_r(x)^4}{\sum_{x}
e^{-\ham_{\sim r}(x)
-\Theta_r(x)^2}}\Big\}\, .
\end{align*}
Since the random variables $e^{-\Theta_r(x)^2}$ and $\Theta_r(x)^4$
are negatively correlated, we have
\begin{align*}
\langle e^{-\Theta_r(x)^2}\Theta_r(x)^4\rangle_{\sim r} \leq \langle
e^{-\Theta_r(x)^2}\rangle_{\sim r} \langle\Theta_r(x)^4 \rangle_{\sim r},
\end{align*}
which implies 
\begin{align}\label{eqn:boundfourthmoment}
\bE\langle (\partial_{rc}\ham)^4\rangle_s {\leq} \frac{64}{\sigma^4}\bE\langle
\Theta_r(x)^4 \rangle_{\sim r}.
\end{align}
Using the inequality $(a+b+c)^4\leq 27(a^4+b^4+c^4)$ and the definition of $\{A_{ij}\}$ and $\{B_{ij}\}$ in
Theorem~\ref{thm:cdmauniversal}, we get
\begin{align*}
\bE\langle (\partial_{rc}\ham)^4\rangle
&\leq \frac{27\cdot 64}{4\sigma^4}\Big\{
\bE[Z_r^4] +\bE[\<(D_{rc}x_c)^4\>_{\sim r}]+
\bE[\<(\sum_{i\in [n]\setminus c}D_{ri}x_i)^4\>_{\sim r}]\Big\}\\
& \le K_1 + K_1\, s^4 + K_1\, \bE[\<(\sum_{i\in [n]\setminus c}D_{ri}x_i)^4\>_{\sim r}\, ,
\end{align*}
where $K=K(\sigma)$ is a constant independent of $m,n$. 
If we use the subscript $i\neq j\neq k\neq\dots$
to denote all the tuples of distinct indices and we expand the power,
we get
\begin{align*}
\bE[\<(\sum_{i\in [n]\setminus c}D_{ri}x_i)^4\>_{\sim r}] &= 
\sum_{i, j, k, l\in [n]\setminus c}
\bE[D_{ri}D_{rj}D_{rk}D_{rl}\<x_ix_jx_kx_l\>_{\sim r}]\\
&=\sum_{i, j, k, l\in [n]\setminus c}
\bE[D_{ri}D_{rj}D_{rk}D_{rl}]\bE[\<x_ix_jx_kx_l\>_{\sim r}]=\\
&= \sum_{i\in [n]\setminus c}
\bE[D_{ri}^4]\bE[\<x_i^4\>_{\sim r}]+
+3\sum_{i\neq j\in [n]\setminus c}\bE[D_{ri}^2D_{rj}^2]
\bE[\<x_i^2x_j^2\>_{\sim r}]\, ,
\end{align*}
Here we used the fact that $\{D_{ri}\}_{1\le i\le n}$ are independent 
of $\ham_{\sim r}(x)$, and therefore of $\<x_ix_jx_kx_l\>_{\sim r}$
Further all the terms with one of the indices $i,j,k,l$
distinct from the all others vanish because $\E\{D_{ri}\}=0$
for all $i\neq c$ by our assumption on $\matA_n$, $\matB_n$.  
Using $x_i\in \{0,2\}$, we then get
\begin{align*}
\bE[\<(\sum_{i\in [n]\setminus c}D_{ri}x_i)^4\>_{\sim r}] &\le
 \sum_{i\in [n]\setminus c} \frac{(1+K)^2}{m^2}\, \cdot 16
+3\sum_{i\neq j\in [n]\setminus c}\frac{1}{m^2}\cdot 16\le K_2
\end{align*}
where $K_2 = K_2(\alpha)$ is another constant. 
Putting everything together, we get 
\begin{align*}
\bE\langle (\partial_{rc}\ham)^4\rangle \le K_3\, (1+s^4)\, .
\end{align*}
and therefore, by Jensen inequality, we get
$\bE\langle|\partial_{rc}\ham|^3\rangle \leq K_3(1+ |s|^3)$
(by eventually enlarging the constant $K_3$. 
Using Eq.~(\ref{eq:FreeEnergyDerivative}), this  finally
implies that
\begin{eqnarray*}
\bE[|\partial^3_{rc}f(\matD_n(r,c,s),Z)|] \leq \frac{K_4}{n} (1+|s|^3)\, .
\end{eqnarray*}

We are now in position to apply Theorem \ref{thm:souravmodified}. 
Since the means and variances of the entries of $\matA_n$ and $\matB_n$
are equal, we have $a_i=b_i=0$. We get therefore
\begin{align*}
|\bE[f(m^{-1/2}\matA_n,Z)] - \bE[f(m^{-1/2}\matB_n,Z)]| &\leq
\frac{K_4}{n}\sum_{r=1}^m\sum_{c=1}^n\big(\bE_{A_{rc}}\int_0^{A_{rc}/\sqrt m} (1+|s|^3)(\frac{A_{rc}}{\sqrt{m}}-s)^2 \de s \\
&\phantom{=====}+ \bE_{B_{rc}}\int_0^{B_{rc}/\sqrt m} (1+|s|^3)(\frac{B_{rc}}{\sqrt{m}}-s)^2 \de s \big)\\
& \leq m K'\sum_{i=3}^6\Big\{\bE\Big[\Big(\frac{A_{rc}}{\sqrt m}\Big)^i\Big]+
\bE\Big[\Big(\frac{B_{rc}}{\sqrt m}\Big)^i\Big]\Big\} = O\Big(\frac{1}{\sqrt n}\Big).
\end{align*}
\end{proof}
The proof of Theorem \ref{thm:cdmasparsedense} is very similar to the one 
above. We only stress the differences below.
\begin{proof}[Proof of Theorem~\ref{thm:cdmasparsedense}]
Let $\matA^\gamma_n$ and $\matB_n$ be as defined in the statement.
We modify the definition of $\matD_n(r,c,s)$ used in the last proof,
as follows
\begin{align*}
D_{ij} = \left\{\begin{array}{ll}
\frac{1}{\sqrt \gamma}A_{ij}^{\gamma},& \text{ if } i< r \text{ or } i=r \text{ and } j< c,\\
s,& \text{ if } i=r, \text{ and } j=c,\\
\frac{1}{\sqrt m}B_{ij},& \text{ otherwise}.
\end{array}\right.
\end{align*}
Now following the proof of Theorem~\ref{thm:cdmauniversal},
and assuming without loss of generality $\gamma\ge 1$, we get 
again
\begin{align*}
\bE\langle (\partial_{rc}\ham)^4\rangle
\leq K_1\Big(1+s^4\Big) \, .
\end{align*}
(The final step consists as in the previous proof, in bounding
the sums $\sum_{i\in [n]\setminus c}\bE[D_{ri}^4]$ and\\
$\sum_{i\neq j\in [n]\setminus c}\bE[D_{ri}^2D_{rj}^2]
\bE[\<x_i^2x_j^2\>_{\sim r}]$.)
This in turn implies
$\bE[|\partial^3_{rc}f(\matD_n(r,c,s),Z)|] \leq (K_1'/n)(1+ |s|^3)$. 
Since the means and variances of the entries of
$\matA^{\gamma}_n$ and $\matB_n$
are equal, we have $a_i=b_i=0$. 
Applying Theorem~\ref{thm:souravmodified}, we get
\begin{align*}
|\bE[f(\gamma^{-1/2}\matA^\gamma_n,Z)] - \bE[f(m^{-1/2}\matB_n,Z)]| &\leq
\frac{K_1'}{n}\sum_{r=1}^m\sum_{c=1}^n\Big\{\bE_{A^\gamma_{rc}}\int_0^{A^\gamma_{rc}/\sqrt\gamma} (1+|s|^3)(\frac{A_{rc}^{\gamma}}{\sqrt{\gamma}}-s)^2 \de s \\
&\phantom{=====}+ \bE_{B_{rc}}\int_0^{B_{rc}/\sqrt m} (1+|s|^3)(\frac{B_{rc}}{\sqrt{m}}-s)^2 \de s \Big\}\\
& \leq m K_2
\sum_{i=3}^6\Big\{\bE\Big[\Big(\frac{A_{rc}^{\gamma}}{\sqrt\gamma}\Big)^i\Big]+
\bE\Big[\Big(\frac{B_{rc}}{\sqrt
m}\Big)^i\Big]\Big\}\\
&\le K_3\Big(\frac{1}{\sqrt\gamma}+\frac{1}{\sqrt n}\Big).
\end{align*}
Now taking the limit $n\to\infty$ first and then the limit $\gamma\to\infty$
gives the result.
\end{proof}
%
%*********************************************************
%
\subsection{LASSO}\label{sec:lassoproof}

The proof of Theorem~\ref{thm:lassouniversal} repeats some arguments
already present in the proof of Theorem~\ref{thm:cdmauniversal} 
presented in the previous section. We shall omit such repetitions and
instead focus on the new ideas required.
\begin{proof}[Proof of Theorem~\ref{thm:lassouniversal}]
Without loss of generality, we will assume 
$x_{\rm max}=1$. Define $\cX = [-1,1]$ and, for $\delta>0$, define 
$\cX_\delta = \{ k\delta\, : \;k\in\integers,
\, |k\delta|\leq 1\}$. In words $\cX_{\delta}$ is a grid of points 
 in the interval $[-1,1]$ with spacing $\delta$. 
Recall that $x_0$ is a fixed deterministic signal with $||x_0||_{\infty}\le 1$,
and the resulting measurements read
$Y=A_nx_0 +Z$, where $Z$ is noise vector with i.i.d. Gaussian 
component. Define the Hamiltonian function
$\ham:\reals^n\times\reals^m\times \reals^{m\times n}\to\reals$
by letting 
\begin{eqnarray*}
\ham(x,z,A_n) &=&\lambda\Vert x\Vert_1+ 
\frac{1}{2} \Vert y - A_n x\Vert_2^2\\
&=&\lambda\Vert x\Vert_1+ 
\frac{1}{2} \Vert z - A_n (x-x_0)\Vert_2^2
\, .
\end{eqnarray*}
 With this
definition, $L(A_n) = \frac{1}{n}\min_{x\in\mathcal{X}^n }\{\ham(x,z,A_n)\}$.
Let $L_{\delta}(A_n) = \frac{1}{n}\min_{x\in\mathcal{X}_\delta^n
}\{\ham(x,z,A_n)\}$.
Our proof follows by first showing that there exists a constant $C$ 
such that 
\begin{align*}
\big|\bE[L_\delta(n^{-1/2}\matA_n)] -\bE[L(n^{-1/2}\matA_n)]\big|
\leq C\, \delta,\;\;
\big|\bE[L_\delta(n^{-1/2}\matB_n)] -\bE[L(n^{-1/2}\matB_n)]\big|
\leq C\, \delta\, .
\end{align*}
Obviously $L_\delta(A_n)\ge L(A_n)$.
In order to prove the converse bound, let $\hx$ be a minimizer of 
$\ham(x,z,A_n)$ in $\mathcal{X}^n$, and denote by $x_\delta$
its closest approximation in $\cX_{\delta}^n$. Obviously 
$|x_{\delta,i}-\hx_i|\leq \delta$ for all $i\in [n]$.
We then have
\begin{align}\label{eqn:diffham}
\frac{1}{n}|\ham(\hx,z,n^{-1/2}\matA_n)-&\ham(x_\delta,z,n^{-1/2}\matA_n)| 
\le  \lambda\delta +
\frac{1}{2n}\Big|\Vert y-n^{-1/2}\matA_n \hx\Vert_2^2 - \Vert y-n^{-1/2}\matA_n
x_\delta\Vert_2^2\Big|\nonumber\\
& = \lambda\delta +\frac{1}{2n} \Big|\big(\frac{1}{\sqrt
n}\matA_n(x_\delta-\hx)\big)^\top 2z +
\big(\frac{1}{\sqrt n}\matA_n(\hx-x_\delta)\big)^\top\frac{1}{\sqrt
n}\matA_n(\hx+x_{\delta}-2x_0)\Big|\nonumber\\
&\leq \lambda\delta +
\frac{1}{2n}\Big\Vert\frac{1}{\sqrt n}\matA_n(\hx-x_\delta)\Big\Vert_2 2
\Vert z\Vert_2 +
\frac{1}{2n^2}\, \sigma_{\rm max}(\matA_n)^2\|\hx-x_{\delta}\|_2
\|\hx+x_{\delta}-2x_0\|_2
\nonumber\\
&\leq \lambda\delta + \sigma_{\rm max}(n^{-1/2}\matA_n)\frac{1}{\sqrt n}
\,\delta\, \Vert z \Vert_2 +
2 \, \sigma_{\rm max}(n^{-1/2}\matA_n)^2\delta\, ,
\end{align}
where we used $\|x_{\delta}\|_2$, $\|\hx\|_2$, $\|x_0\|_2\le \sqrt{n}$ and
$\|\hx-x_{\delta}\|_2\le \delta\sqrt{n}$. 
Here $\sigma_{\rm max}(A_n)$ is the largest singular value of $A_n$. 
%From \cite{YinBai} we know that $\sigma_{\rm max}(n^{-1/2}\matA_n) 
%\to 1+\sqrt{\alpha}$ almost surely. 
From \cite{Seginer} we know that $\bE[\sigma_{\max}(n^{-1/2}\matA_n)^2] < K$, for some
constant $K<\infty$. 
% Moreover,
%under the bounded sixth moment condition, we also have
%$\sigma_{\rm max}(n^{-1/2}\matA_n) \to 1+\sqrt{\alpha}$ in $L_1$.
Combining this with Eq.~\eqref{eqn:diffham}, and 
using the Cauchy-Schwartz inequality, 
%$\Vert z \Vert_2\le 2\sqrt{n}$ with probability at least $1-e^{-An}$,
we get
\begin{align*}
\big|\bE[L_\delta(n^{-1/2}\matA_n)] -\bE[L(n^{-1/2}\matA_n)]\big| 
\leq C\delta\, .
\end{align*}
A similar result obviously holds for the matrix ensemble 
$\matB_n$ as well. By triangular inequality, we have
\begin{align*}
\lim_{n\to\infty}\big|\bE[L(n^{-1/2}\matA_n)] -\bE[L(n^{-1/2}\matB_n)]\big| 
\le 
C\delta + \lim_{n\to\infty}\big|\bE[L_\delta(n^{-1/2}\matA_n)] -
\bE[L_\delta(n^{-1/2}\matB_n)]\big|\, .
\end{align*}
Since this inequality holds for any $\delta>0$, the proof of the theorem
reduces to showing that 
$\lim_{n\to\infty}\big|\bE[L_\delta(n^{-1/2}\matA_n)] -
\bE[L_\delta(n^{-1/2}\matB_n)]\big| = 0$.

In order to prove this, define
\begin{align*}
f(\delta, \beta,z,A_n) &= -\frac{1}{\beta n}\log
\Big\{\sum_{x\in\mathcal{X}_\delta^n}
e^{-\beta \ham(x,z,A_n)}\Big\}\, .
\end{align*}
It is easy to see that
\begin{align}\label{eqn:zerotemp}
\lim_{\beta\to\infty}f(\delta,\beta,z,A_n) = \frac{1}{n}\min_{x \in
\mathcal{X}_\delta^n} \ham(x,z,A_n)\, .
\end{align}
Further, a straightforward calculation shows that 
\begin{eqnarray*}
\beta^2 \frac{\partial f}{\partial \beta}(\delta, \beta, z, A_n)
=H(p_{\beta,A_n})\, ,
\end{eqnarray*}
where $H(p)$ denotes Shannon's entropy of the probability distribution 
$p$ and $p_{\beta,A_n}(x) \propto \exp\{-\beta\ham(x,z,A_n)\}$.
Of course $0\le H(p_{\beta,A_n})\le n\log |\cX_{\delta}|$ whence
\begin{align} \label{eqn:convexity}
-\frac{1}{\beta^2}\log\Big(\frac{2}{\delta}\Big)\le 
\frac{\partial f}{\partial \beta}(\delta, \beta, z, A_n)
\le 0\, .
\end{align}
Therefore, 
\begin{align}\label{eqn:deltabound}
\lim_{n\to\infty}&\big|\bE[L_\delta(n^{-1/2}\matA_n)] -
\bE[L_\delta(n^{-1/2}\matB_n)]\big| \nonumber\\
&\stackrel{}{=}
\lim_{n\to\infty}\big|\lim_{\beta\to\infty}\bE[f(\delta,\beta,Z,n^{-1/2}\matA_n)]
-\lim_{\beta\to\infty}\bE[f(\delta,\beta,Z,n^{-1/2}\matB_n)]\big|\nonumber\\
&\stackrel{}{\leq}\lim_{n\to\infty}\big|\bE[f(\delta,\beta,Z,n^{-1/2}\matA_n)]
-\bE[f(\delta,\beta,Z,n^{-1/2}\matB_n)]\big|
+\int_{\beta}^\infty\frac{1}{s^2}\log\Big(\frac{2}{\delta}\Big)\,\, \de s\, ,
\end{align}
where the first step follows from \eqref{eqn:zerotemp}
and the second from \eqref{eqn:convexity}.
Notice the close resemblance between the function $f(\delta,\beta,Z,A_n)$
defined here and the one used in the previous section.
Using the same arguments developed there for the proof
of Theorem \ref{thm:cdmauniversal} it is immediate to show that
\begin{align*}
\big|\bE[f(\delta,\beta,Z,n^{-1/2}\matA_n)]-\bE[f(\delta,\beta,Z,n^{-1/2}
\matB_n)]\big|
\leq O\Big(\frac{1}{\sqrt n}\Big)\, .
\end{align*}
Combining this with Eq.~\eqref{eqn:deltabound}, we get
\begin{align*}
\lim_{n\to\infty}&\big|\bE[L_\delta(n^{-1/2}\matA_n)] -
\bE[L_\delta(n^{-1/2}\matB_n)]\big| \leq\frac{1}{\beta} \log\Big(
\frac{2}{\delta}\Big)\, .
\end{align*}
The proof is completed by letting $\beta\to\infty$.
\end{proof}
%
%*********************************
%
\subsection{Wishart Matrices}\label{sec:wishartproof}

The proof is analogous the proof for universality of the
Wigner's semi-circle law developed in \cite{Sou05unpublished}. 
\begin{proof}[Proof of Theorem~\ref{thm:wishartsparsedense}]
By the analiticity of the Stieltjes transform, it is sufficient to prove the 
claim for ${\rm Im}(z)$ large enough.

For an $m\times n$ matrix $A_n$ and any $z\in\complex\backslash \reals$, let 
\begin{align*}
f(A_n) \equiv \frac1n\trace\big((A_n^\top A_n + zI_n)^{-1}\big)\, .
\end{align*}
In order to simplify the notation we drop the subscript $n$ and denote
the partial derivative with respect to $A_{ij}$ by $\partial_{ij}$.
Define $R=  (A^\top A + zI)^{-1}$. Therefore $(A^\top A+
zI)R = I$, which implies $\partial_{ij}((A^\top A+
zI)R) = 0$. This yields 
\begin{align*}
\partial_{ij} R = -R\partial_{ij}(A^\top A) R\, .
\end{align*}
Let $\matone_{ij}$ 
denote the matrix with $(ij)$-th entry equal to $1$ and the remaining
entries equal to $0$. Then 
\begin{align*}
\partial_{ij}( A^\top A )&= \matone_{ji} A + A^\top\matone_{ij}, \\
\partial_{ij}^2( A^\top A ) &= 2 \matone_{ii}\, ,\\
\partial_{ij}^3( A^\top A ) & = 0\, .
\end{align*}
Using the identity $\trace(AB) = \trace(BA)$, we get
\begin{align}\label{eqn:partial3}
\partial_{ij} f &= -\frac{1}{n}\trace\Big(\partial_{ij}(A^\top A)
R^2\Big),\nonumber\\
\partial^2_{ij} f &=
\frac{2}{n}\trace\Big(\partial_{ij}(A^\top A) R\partial_{ij}(A^\top A)R^2\Big)
- \frac{1}{n}\trace\Big(\partial_{ij}^2(A^\top A)R^2\Big),\nonumber\\
\partial^3_{ij} f &= -\frac{6}{n} \trace\Big(
\partial_{ij}(A^\top A)R\partial_{ij}(A^\top A)R\partial_{ij}(A^\top A)R^2\Big)
\phantom{=} +\frac{3}{n}\trace\Big(\partial^2_{ij}(A^\top A)R\partial_{ij}(A^\top A)R^2\Big)
\nonumber\\
&+
\frac{3}{n}\trace\Big(\partial_{ij}(A^\top A)R\partial^2_{ij}(A^\top A)R^2\Big).
\end{align} 
Note that $R$ is a symmetric matrix and therefore
is diagonalizable. Moreover, note that the
singular values of $R^{-1}$ are bounded by $|v|^{-1}$, where $v = \text{Im}(z)$.
Let $\Vert A\Vert$ and $\Vert A \Vert_2$ denote the Frobenius norm and
the spectral norm of $A$ respectively. From Cauchy-Schwartz
inequality we have $|\trace(AB)|\leq \Vert A\Vert\,\Vert B\Vert$.
Therefore, we can bound the first term as
\begin{align}\label{eqn:trace1}
|\trace(\partial_{ij}(A^\top A)R\partial_{ij}(A^\top A)R\partial_{ij}(A^\top A)R^2)|
&{\leq}
\Vert(\partial_{ij}(A^\top A)R)^2\Vert\Vert\partial_{ij}(A^\top A)R^2\Vert\nonumber\\
&\stackrel{(a)}{\leq}\frac{1}{|v|}\, \Vert\partial_{ij}(A^\top A)R\Vert^3\nonumber\\
&\stackrel{(b)}\leq\frac{1}{|v|^4}\, \Vert\partial_{ij}(A^\top A)\Vert^3\, ,
\end{align}
where we have used 
$\Vert A B\Vert\leq\Vert A\Vert\Vert B\Vert$ in $(a)$ and
$\Vert A B\Vert\leq \Vert A\Vert\Vert B\Vert_2$ in both $(a)$ and $(b)$.

Similarly one can bound the second and third terms of \eqref{eqn:partial3} as
\begin{align}\label{eqn:trace2}
|\trace(\partial^2_{ij}(A^\top A)R\partial_{ij}(A^\top A)R^2)|
\leq
\Vert\partial_{ij}^2(A^\top A)\Vert\Vert\partial_{ij}(A^\top A)\Vert\frac{1}{|v|^3}
= \Vert\partial_{ij}(A^\top A)\Vert\frac{2}{|v|^3}. 
\end{align}
Finally, we can bound $\Vert\partial_{ij}(A^\top A)\Vert$ as follows
\begin{align}\label{eqn:bound1}
\Vert\partial_{ij}(A^{\top} A)\Vert \leq \Vert\matone_{ji}A\Vert +
\Vert A^\top\matone_{ij}\Vert = 2 \Vert A^\top\matone_{ij}\Vert =
2\Big(\sum_{k=1}^mA_{kj}^2\Big)^{1/2}\, .
\end{align}

Let us now consider the random matrices $\matA^\gamma_n$ and $\matB_n$ as defined in
the theorem. Let $\matC_n(r,c,s)$ denote the matrix as defined in
Section~\ref{sec:cdmaproof}, i.e.,
\begin{align*}
C_{ij} = \left\{\begin{array}{ll}
\frac{1}{\sqrt \gamma}A^\gamma_{ij},& \text{ if } i< r \text{ or } i=r \text{ and } j< c,\\
s,& \text{ if } i=r, \text{ and } j=c,\\
\frac{1}{\sqrt m}B_{ij},& \text{ otherwise}.
\end{array}\right.
\end{align*}
Using the equations \eqref{eqn:partial3},
\eqref{eqn:trace1}, \eqref{eqn:trace2}, and \eqref{eqn:bound1}, we get
\begin{align*}
\bE\big\{|\partial_{rc}^3f(\matC_n(r,c,s))|\big\} &\leq
\frac{K_0}{n}\E\Big\{\Big(1+\sum_{k=1}^mC_{kc}^2\Big)^{3/2}\Big\}\\
&\le \frac{K_1}{n}
(1+s^3)\, .
\end{align*}
The  proof is finished as for Theorem~\ref{thm:cdmasparsedense}.
\end{proof}

%
%*************************************************************************
%
\subsection{Proof of Corollary~\ref{cor:mimosparsedense}}

Throughout this proof we will assume $\sigma=1$, for simplicity
of notation (general $\sigma>0$ follows exactly the same argument).

Convergence of Stieltjes transform implies weak convergence
of the expected distribution of eigenvalues \cite[Theorem~2.4.4]{Guionnet}.
This means  that for any continuous bounded function $f$\footnote{Note that we
have two limits on the left hand side. This can be taken care of by noticing that 
$\lim_{\gamma\to\infty}\lim_{n\to\infty}f(\gamma,n,x) = f(x)$ is equivalent 
to saying that $\lim_{n\to\infty}f(\gamma_n,n,x)=f(x)$ along any sequence of $\{\gamma_n\}$s
satisfying $\lim_{n\to\infty}\gamma_n = \infty$.} 
\begin{align}
\lim_{\gamma\to\infty}\lim_{n\to\infty}\frac{1}{n} \sum_{i=1}^n
\bE[f(\lambda_i(\gamma^{-1}(\matA^\gamma_n)^\top\matA_n^\gamma))] = 
\lim_{n\to\infty}
\frac{1}{n}\sum_{i=1}^n
\bE[f(\lambda_i(n^{-1}\matB_n^\top\matB_n))]\, .\label{eq:WeakConvergence}
\end{align}
The limit on the right hand side exists because the expected
distribution of eigenvalues of Wishart matrices converges 
\cite{MarPas,BaiBook}.
Moreover, the limiting distribution function is continuous. Therefore, 
the convergence of the distributions implies the convergence of expectations
for any bounded measurable function, not necessarily continuous
(by the bounded convergence theorem).
We are interested in estabilishing a result of the form 
(\ref{eq:WeakConvergence}) for the function $f(x) = \log(1+x)$, which is 
not bounded. 
However, note that only the behavior of
$f$ in the region $x\geq 0$ is relevant, because $\lambda_i\geq 0$. In
the domain of interest the function $f$ is bounded from below. In order to tackle the
issue of boundedness from above, we use a standard truncation trick.
We define $g_M(x) = f(x)\indicator{x\leq M}$, for some $0< M < \infty$. Note
that the function $g_M$ is bounded on $\reals_+$. Therefore
\begin{align}\label{eqn:boundedsparsedense}
\lim_{\gamma\to\infty}\lim_{n\to\infty}\frac{1}{n} \sum_{i=1}^n
\bE[g_M(\lambda_i(\gamma^{-1}(\matA^\gamma_n)^\top\matA_n^\gamma))] = 
\lim_{n\to\infty}
\frac{1}{n}\sum_{i=1}^n
\bE[g_M(\lambda_i(n^{-1}\matB_n^\top\matB_n))]\, .
\end{align}
Note that 
\begin{align}\label{eqn:sparseboundedunbounded}
\lim_{\gamma\to\infty}\lim_{n\to\infty}&\frac{1}{n} \sum_{i=1}^n
\bE\Big\{\big|g_M(\lambda_i(\gamma^{-1}(\matA^\gamma_n)^\top\matA_n^\gamma))-
f(\lambda_i(\gamma^{-1}(\matA^\gamma_n)^\top\matA_n^\gamma))\big|\Big\}
\nonumber\\
&= \lim_{\gamma\to\infty}\lim_{n\to\infty}\frac{1}{n} \sum_{i=1}^n
\bE\Big\{\log(1+\lambda_i(\gamma^{-1}(\matA^\gamma_n)^\top\matA_n^\gamma))
\indicator{\lambda_i> M}\Big\} \nonumber\\
&\leq \lim_{\gamma\to\infty}\lim_{n\to\infty}\frac{1}{n} \sum_{i=1}^n
\bE\Big\{\lambda_i^2(\gamma^{-1}(\matA^\gamma_n)^\top\matA_n^\gamma))/M\Big\}\nonumber\\
&=\lim_{\gamma\to\infty}\lim_{n\to\infty} \frac{1}{Mn\gamma^2}\bE\,\trace 
\Big\{\Big((A_n^\gamma)^\top
A_n^\gamma \Big)^2\Big\}\nonumber\\
& = \lim_{\gamma\to\infty}\lim_{n\to\infty} \frac{1}{Mn\gamma^2}
\bE\Big\{\sum_{i,j}\Big(\sum_{k} A_{ki}A_{kj}\Big)^2\Big\}\nonumber\\
& =  \lim_{\gamma\to\infty}\lim_{n\to\infty} \frac{1}{Mn\gamma^2}\Big\{\sum_{i\neq
j}\sum_{k}\bE\big\{A_{ki}^2\big\} \bE\big\{A_{kj}^2\big\} + \sum_{i}
\sum_{k_1\neq k_2}\bE\big\{A_{k_1i}^2\big\}\bE\big\{A_{k_2i}^2\big\} +
\sum_{i,k}\bE\big\{A_{ki}^4\big\}\Big\}\nonumber\\
& \le \frac{K}{M}\, ,
\end{align}
for a constant $K$ independent of $M$, $\gamma$ as long as $\gamma\ge 1$.
Using a similar argument we can show that
\begin{align}\label{eqn:denseboundedunbounded}
\lim_{n\to\infty}
\frac{1}{n}\sum_{i=1}^n
\bE\big\{|g_M(\lambda_i(n^{-1}\matB_n^\top\matB_n))]
-f(\lambda_i(n^{-1}\matB_n^\top\matB_n))|\}\ \leq \frac{K'}{M}
\end{align}
for a constant $K'$ independent of $M$.
From \eqref{eqn:boundedsparsedense}, \eqref{eqn:sparseboundedunbounded},
\eqref{eqn:denseboundedunbounded} we get 
\begin{align*}
\lim_{\gamma\to\infty}\lim_{n\to\infty}\big|\bE[C_n(\gamma^{-1/2}\matA^\gamma_n)]
- \bE[C_n(n^{-1/2}\matB_n)]\big| \leq \frac{K+K'}{M}\, .
\end{align*}
%for $M > (1+\sqrt \alpha)^2$. 
Now taking the $\lim_{M\to\infty}$ gives the desired result.

\endproof

\bibliographystyle{amsalpha}

\begin{thebibliography}{99}

\bibitem{Sou06}
S.~Chatterjee, ``A generalization of the {Lindeberg} principle,'' \emph{The
  Annals of Probability.}, vol.~34, no.~6, pp. 2061--2076, 2006.

%\bibitem{Pas72}
%L.~A. Pastur, ``On the spectrum of random matrices,'' \emph{TMF}, vol.~10,
%  no.~1, pp. 102--112, 1972.

\bibitem{MarPas}  V.~A.~Mar{\v{c}}enko and L.A.~Pastur,
``Distribution of Eigenvalues for Some Sets of Random Matrices'',
Math. USSR Sb. 1 (1967) 457-483 


\bibitem{StatMech} J.~Cardy, \emph{Scaling and Renormalization in Statistical
Physics}, Cambridge University Press, Cambridge, 1996


\bibitem{Tel99}
E.~Telatar, ``Capacity of multi-antenna {Gaussian} channels,'' \emph{European
  Transactions on Telecommunications}, vol.~10, no.~6, pp. 585--595, 1999.

\bibitem{DoTUniversality} D.~L. Donoho and J.~Tanner,
``Observed Universality of Phase Transitions in High-Dimensional Geometry,
with Implications for Modern Data Analysis and Signal Processing,''
Phil. Trans. R. Soc. A 13, November 2009, 4273-4293

\bibitem{DoT09}
D.~L. Donoho and J.~Tanner, ``Counting faces of randomly-projected polytopes
  when the projection radically lowers dimension,'' \emph{Journal of the AMS},
  vol.~22, no.~1, pp. 1--53, 2009.

%\bibitem{DMM09}
%D.~L. Donoho, A.~Maleki, and A.~Montanari, ``Message passing algorithms for
%  compressed sensing,'' \emph{CoRR}, vol. abs/0907.3574, 2009.

\bibitem{Kesten} H.~Kesten, ``Symmetric random walks on groups,''
\emph{Trans. Amer. Math. Soc.} 92 (1959), 336-354

\bibitem{McKay} B.~D.~McKay, ``The expected eigenvalue distribution of a large 
regular graph''\emph{Linear Algebra Appl.} 40 (1981) 203-216

\bibitem{VerduBook} S.~Verdu, \emph{Multiuser Detection},
Cambridge University Press, Cambridge, 1998

\bibitem{GrA96}
A.~J. Grant and P.~D. Alexander, ``Randomly selected spreading sequences 
for  coded CDMA,'' in \emph{4th Int. Spread Spectrum 
Techniques and Applications},  Mainz, Germany, Sept. 1996, pp. 54--57.

\bibitem{KoM08}
S.~B. Korada and N.~Macris, ``Tight bounds on the capacity of binary input
  random {CDMA} systems,'' \emph{accepted in IEEE Trans. Inform. Theory}
{\sf arXiv:0803.1454} (2008)

\bibitem{MoT06itw}
A.~Montanari and D.~Tse, ``Analysis of belief propagation for non-linear
  problems: The example of {CDMA} (or : {H}ow to prove {T}anaka's formula),''
  in \emph{Proc. of the IEEE Inform. Theory Workshop}, Punta del Este, Uruguay,
  Mar 13--Mar 17 2006.

\bibitem{Tanaka} T.~Tanaka,
``A statistical-mechanics approach to large-system analysis of CDMA
multiuser detectors,'' IEEE Trans. on Inform. Theory, 48 (2002),
2888-2910

\bibitem{Tibshirani} R.~Tibshirani,
``Regression shrinkage and selection via the LASSO,''
J. Roy. Statist. Soc. B 58 (1996), 267--288

\bibitem{ChenDonoho} S.S. Chen, D.L.~Donoho, and M.A.~Saunders,
``Atomic decomposition by basis pursuit,'' SIAM Review 43 (2001),
129--159

\bibitem{Guionnet} G.~W.~Anderson, A.~Guionnet, and O.~Zeitouni,
\emph{An Introduction to Random Matrices},
Cambridge University Press, Cambridge, 2009

\bibitem{MM} M.~M\'ezard and A.~Montanari,
\emph{Information, Physics, and Computation}, Oxford University Press,
Oxford 2009

\bibitem{Tala} M.~Talagrand,
\emph{Spin Glasses: A Challenge for Mathematicians: Cavity and Mean 
Field Models}, Springer, New York, 2003

\bibitem{GuerraToninelli} F.~Guerra and F.~L.~Toninelli,
``The High Temperature Region of the Viana-Bray Diluted Spin Glass Model,''
J.~Stat.~Phys. 115 (2004) 531-555


\bibitem{Tal02}
M.~Talagrand, ``Gaussian averages, Bernoulli averages, and Gibbs' measures,''
\emph{Random Structures and Algorithms}, vol.~21, no.~3-4, pp. 197--204, 2002.

\bibitem{Ton02}
F.~L. Toninelli, ``Rigorous results for mean field spin glasses: thermodynamic
  limit and sum rules for the free energy,'' Ph.D. dissertation,
 Scuola Normale Superiore, Pisa, Italy, 2002.

\bibitem{Seginer}
Y. Seginer, ``The Expected Norm of Random Matrices", {\em Combinatorics, Probability
and Computing}, vol.~9 pp. 149-166.

\bibitem{Sou05unpublished}
S.~Chatterjee, ``A simple invariance theorem,'' \emph{unpublished.}, vol.
  http://arxiv.org/abs/math/0508213v1, 2005.

\bibitem{BaiBook} Z.~Bai and J.~W.~Silverstein
\emph{Spectral Analysis of Large Dimensional Random Matrices},
Springer, New York, 2009

\end{thebibliography}

\end{document}